\documentclass[letter,12pt]{amsart}
\usepackage{amsmath, amsthm, amstext}
\usepackage{amssymb, array, amsfonts}
\usepackage[english]{babel}
\usepackage{bbding}
\usepackage{float}
\usepackage{tikz}
\usepackage{pdfpages}
\usepackage{hyperref}
\usepackage{float}
\usepackage{graphics}
\usepackage{graphicx}
   \usepackage{enumerate}
\usepackage[margin=2.2cm]{geometry}
\numberwithin{equation}{section}

\def \ie {i.\,e.\,}

\def \bark{{\bar{k}}}

\def \er{\varepsilon}

\def \ka{\varkappa}
\def \H{\mathrm{H}}

\def \la{\lambda}

\renewcommand{\l}{\left}
\renewcommand{\div}{\operatorname{div}}
\newcommand{\per}{\operatorname{per}}
\renewcommand{\r}{\right}

\def \D{\Delta}
\def \O{\mathcal O}

\def \C{\mathbb{C}}
\def \c{\mathrm{C}}

\def \N{\mathbb{N}}
\def \M2{\mathrm{M}_2}
\def \R{\mathbb{R}}
\def \Z{\mathbb{Z}}

\def \H{\mathrm{H}}
\def \L{\mathrm{L}}
\def \sl2r{\mathrm{SL}(2,\R)}

\def\dd{\partial}

\newcommand{\beq}{\begin{equation}}
\newcommand{\eeq}{\end{equation}}
\newcommand{\one}{\mathbf{1}}
\renewcommand{\D}{\mathcal{D}}

\def\dom{\operatorname{Dom}}
\def\eff{\operatorname{eff}}

\def\wt{\widetilde}

\def\wt{\widetilde}
\def\ran{\operatorname{ran}}

\def\im{\operatorname{Im}}
\def\re{\operatorname{Re}}
\newcommand{\eqdef}{\stackrel{\rm def}{=\kern-3.6pt=}}
    \newcommand{\<}{\langle}
\renewcommand{\>}{\rangle}

\theoremstyle{plain}
\newtheorem{theorem}{\bf Theorem}[section]
\newtheorem{lemma}[theorem]{\bf Lemma}

\newtheorem{prop}[theorem]{\bf Proposition}
\newtheorem{cor}[theorem]{\bf Corollary}

\theoremstyle{definition}

\theoremstyle{remark}
\newtheorem{remark}[theorem]{\bf Remark}

\theoremstyle{cond}

\renewcommand{\le}{\leqslant}
\renewcommand{\ge}{\geqslant}
\newcommand{\rank}{\mathop{\mathrm{rank}}\nolimits}
\newcommand{\dist}{\mathop{\mathrm{dist}}\nolimits}

\renewcommand{\qed}{\vrule height7pt width5pt depth0pt}
\title{On the structure of band edges of 2d periodic elliptic operators}
\author[N. Filonov]{Nikolay Filonov}
\address{St.~Petersburg Department of V.~A.~Steklov Mathematical Institute,
	Fontanka 27,
 St.Petersburg, 191023, Russia, and
	St. Petersburg State University,
Universitetskaya emb. 7/9, St. Petersburg, 199034,
Russia}
\thanks{The first author was supported by RFBR Grant 16--01--00087 and by Simons Foundation.}
\author[I. Kachkovskiy]{Ilya Kachkovskiy}
\address{Department of Mathematics,
	Michigan State University,
	Wells Hall, 619 Red Cedar Road,
	East Lansing, MI, 48910,
	United States of America}
	\thanks{The second author was supported by AMS Simons Travel Grant 2014 -- 2016 and by NSF grant DMS--1758326.}
\dedicatory{To the memory of Yuri Safarov, our dear friend and colleague}
\date{}
\begin{document}
\begin{abstract}
    For a wide class of 2D periodic elliptic operators,
    we show that the global extrema of all spectral band functions are isolated.
    
\end{abstract}
\maketitle
\section{Introduction}
The structure of band edges of periodic Schr\"odinger operators is
an interesting and wide open question of mathematical physics. For
example, suppose that a band function $k\mapsto E(k)$ has a
minimum (or maximum) $k_0$. In solid state physics, the {\it
tensor of effective masses} $M_{\eff}$ at $k_0$ is defined as
\beq
\label{eff_mass}
\l\{M_{\eff}^{-1}\r\}_{ij}=\pm\frac{1}{\hbar^2}
\l.\frac{\dd^2E}{\dd k_i \dd k_j}\r|_{k=k_0}
\eeq
(see \cite[Chapter 12, (12.29)]{AMe} for more details). The choice
of the sign depends on whether the extremum is a minimum (``$+$'',
the effective mass of an electron) or a maximum (``$-$'', the
effective mass of a hole). This definition of $M_{\eff}$ makes
sense only if the right hand side is invertible, \ie if the
critical point $k_0$ is non-degenerate. This is always true in one
dimension, see, for example, \cite[Section XIII.16]{RS4}. It is
commonly believed that, for $d\ge 2$, the spectral gap edges are
non-degenerate for ``generic'' potentials, see, for example,
\cite[Conjecture 5.1]{KuP} and the recent review \cite[Section 5.9.2]{Ku_review}.
However, there are very few rigorous
results in this direction. In \cite{KSi}, it is shown that the
lowest eigenvalue for the periodic Schr\"odinger operator is
non-degenerate. The same holds for the two-dimensional Pauli
operator, see \cite{Pauli}. A wide class of operators for which
the lower edge of the spectrum can be extensively analysed is
described in \cite{BSu03}. See also the survey \cite{Ku} on
photonic crystals, where additional references are given. 
For periodic {\it magnetic} Schr\"odinger operators, 
even the lowest
eigenvalue may be degenerate (i. e. the right hand side of \eqref{eff_mass} may vanish, see \cite{Sht2}). Note, however, that this
can happen only for sufficiently large magnetic potentials, as shown in
\cite{Sht3}.

Much less is known about the edges of other bands. In \cite{KL}, it is established that, for periodic operators of the form $-\Delta+V$ with generic $V$, the edge of each spectral gap is an 
extremum of only one band function, but the question of non-degeneracy of these
extrema remains open. In \cite{V}, it is shown in 2D that for any $N$ there exists a $C^{\infty}$-neighbourhood of $0$ such that, for potentials $V$ from a dense $G_{\delta}$-subset of that neighbourhood, the first $N$ band functions are Morse functions. In other words, any finite number of bands is non-degenerate for generic  $C^{\infty}$-small potentials.

In the present paper, we establish the following result (Theorem
\ref{main}): for a wide class of 2D periodic elliptic second order
operators, any global minimal or maximal value of any band function can only be attained at a discrete set of points.
 In other words, {\it the global extrema of each band function are isolated}. In particular, this implies that the level sets corresponding to spectral band edges cannot contain 1D curves. We do not need any
genericity or smallness assumptions, and our result holds for all
bands, not necessarily for the edges of the spectrum. We formulate
the results for ``smooth'' second order elliptic operators
\eqref{h_def}. We believe that, using methods from \cite{Sht1},
the result can be extended to the same generality in which the
absolute continuity of the spectrum in 2D is established. The
extension beyond dimension 2, however, seems significantly more
challenging, as our technique relies heavily on 2D specifics.

An immediate consequence of our result is that Liouville theorems
(in the sense of \cite{KuP2,KuP}) hold for the operator
\eqref{h_def} at all gap edges, see Corollary \ref{liouville}. Our
result can also be used in studying Green's function asymptotics
near spectral gap edges, see \cite{KKR,KR,Kha}, and to obtain a ``variable period'' version of the non-degeneracy conjecture in 2D \cite{PS}.

Surprisingly, the statement of the main theorem {\it fails for
discrete periodic Schr\"odinger operators on $\Z^2$}, already in
the case of a diatomic lattice. We explain
the corresponding example of non-isolated  extrema in Section 7. 
\vskip 1mm {\noindent \bf
Acknowledgements.} The results were partially obtained during the
programme Periodic and Ergodic Spectral Problems in January --
June 2015, supported by EPSRC Grant EP/K032208/1. We are grateful to the Isaac Newton Institute for
Mathematical Sciences, Cambridge, for their support and
hospitality. We would like to express our deepest thanks to Peter
Kuchment, Leonid Parnovski, and Roman Shterenberg for several
lively and fruitful discussions during the programme, and to Alexander Pushnitski for reading the draft version of the text and for valuable remarks. We are also grateful to the anonymous referee, whose suggestions significantly improved the quality of the paper and the references.
\section{Main result}
Let
$$
\Gamma=\{n_1 b_1+n_2 b_2,n_1,n_2\in \Z\}
$$
be a lattice in $\R^2$, and let $\Omega\subset \R^2$ be an elementary cell of
$\Gamma$ identified with $\R^2/\Gamma$. We will use notation such
as $\c^1_{\per}(\Omega)$, $\H^1_{\per}(\Omega)$ for the classes of
functions satisfying periodic boundary conditions.

The periodic magnetic Schr\"odinger operator with metric $g$ is
defined by the expression \beq \label{h_def}
(Hu)(x)=(-i\nabla -A(x))^{*}g(x)(-i\nabla -A(x))u(x)+V(x)u(x),
\eeq where the electric potential $V\colon \R^2\to \R$ is
$\Gamma$-periodic, \ie assumed to satisfy \beq
\label{v_cond} V(x+b_j)=V(x),\quad j=1,2,\quad V\in
{\L}^{\infty}(\Omega), \eeq and the magnetic potential $A\colon
\R^2\to \R^2$ is also $\Gamma$-periodic and \beq \label{a_cond}
\quad A\in \c_{\per}^1(\Omega;\R^2), \quad \div A=0, \quad
\int_{\Omega}A(x)\,dx=0. \eeq 
Note that the last two conditions
can be imposed without loss of generality, see Remark \ref{gauge}. The metric $g$ is a $\Gamma$-periodic
symmetric $(2\times 2)$-matrix function satisfying \beq
\label{g_cond} g\in \c_{\per}^2(\Omega;\mathrm{M_2}(\R)), \quad
g(x)\ge m_g\one>0,\quad\text{where}\,\, \one=\begin{pmatrix}
    1&0\\0&1
\end{pmatrix},
\eeq for some positive constant $m_g$. The operator \eqref{h_def}
is self-adjoint on $\L^2(\R^2)$ with the domain being
the Sobolev space $\H^2(\R^2)$. From the standard Floquet--Bloch
theory (see, for example, \cite[Section XIII.16]{RS4} or \cite[Section 4.5]{Ku_book}), it follows that
$H$ is unitarily equivalent to the direct integral \beq
\label{dirint} \int_{\wt \Omega}^{\oplus} H(k)\,dk, \eeq where
$\wt \Omega\in\R^2$ is an elementary cell $\R^d/\Gamma'$ of the dual lattice
\beq
\label{lastp2}
\Gamma'=\{m_1 b_1'+m_2 b_2',\,m_1,m_2\in 2\pi \Z\},\quad \<b_i,b_j'\>=\delta_{ij},
\eeq
and the ($m$-sectorial) operators $H(k)$ in
$\L^2(\Omega)$ are defined on the domain $\H^2_{\per}(\Omega)$ by \beq
\label{hk_def} H(k)=(-i\nabla+\bark-A)^*g(-i\nabla+k-A)+V,\quad
k\in \C^2. \eeq 
The family \eqref{hk_def} is an analytic
type A operator family with a compact resolvent, in the sense of \cite{Kato}. This means
that the domains $\dom H(k)$ do not depend on $k$, and $H(k)u$ is
a (weakly) analytic vector-valued function of $k_1$ and $k_2$ for
any $u\in \dom H(k)=\H^2_{\per}(\Omega)$. Note that, while \eqref{dirint} and the statement of the main result only use real values of $k$, we will often need to consider \eqref{hk_def} for $k\in \C^2$, and we need $\bark$ in the definition to keep the expression analytic.

For $k\in \R^2$, let us denote the eigenvalues of $H(k)$, taken in the
non-decreasing order, by $\lambda_j(k)$. These eigenvalues, considered
as functions of $k$, are called {\it band functions}. These
functions are $\Gamma'$-periodic and piecewise real analytic on
$\R^2$. The spectrum of $H$
$$
\sigma(H)=\bigcup_j [\lambda_j^-,\lambda_j^+]
$$
is the union of the {\it spectral bands}
$[\lambda_j^-,\lambda_j^+]$ which are the ranges of
$\lambda_j(\cdot)$. It is well known (under much wider assumptions
than ours, see \cite{BSu,Sht1}) that there are no degenerate
bands, \ie we always have $\lambda_j^-<\lambda_j^+$ . The bands,
however, can overlap. Our main result concerns the structure of the extrema of band functions.
\begin{theorem}
\label{main}Let $H$ be the operator $\eqref{h_def}$ with the
potentials and the metric satisfying $\eqref{v_cond}$,
$\eqref{a_cond}$, $\eqref{g_cond}$. Let $\lambda_*$ be a global minimal or maximal value of $\lambda_j(\cdot)$. Then the level set
$$
\{k\in \wt\Omega \colon \lambda_j(k)=\lambda_*\}
$$
is finite.
\end{theorem}

The following Liouville theorem at the edge of the spectrum
follows immediately from Theorem \ref{main}, see \cite[Theorem 23
and Remark 6.1]{KuP2} or \cite[Theorem 4.4]{KuP}.
\begin{cor}
\label{liouville} Let $H$ be the operator $\eqref{h_def}$ with the
potentials and the metric satisfying $\eqref{v_cond}$,
$\eqref{a_cond}$, $\eqref{g_cond}$. Then for every fixed $\lambda_*\in
\dd(\sigma(H))$ and $n\in \N$, the space of solutions of
$$
(-i\nabla -A(x))^{*}g(x)(-i\nabla -A(x))u(x)+V(x)u(x)=\lambda_*
u(x)
$$
satisfying
$$
|u(x)|=O((1+|x|)^n)
$$
has finite dimension $($which may depend on $\lambda_*$ and $n)$.
\end{cor}
\begin{remark}
\label{gauge}
The second and third conditions from \eqref{a_cond} can be imposed without loss of generality using a gauge transformation $A\mapsto A-\nabla\Phi-|\Omega|^{-1}\int_\Omega A(x)\,dx$ (see, for example, \cite[Section 1.2]{BSu98}) with $\Phi$ periodic. The addition of $-\nabla\Phi$ is a unitary equivalence transformation of $H(k)$ for all $k$, and the addition of the last term is equivalent to the change of the quasimomentum $k\mapsto k-|\Omega|^{-1}\int_\Omega A(x)\,dx$. Neither of these changes affects the main result.
\end{remark}
\vskip 3mm {\noindent \bf The structure of the paper. }In Sections
3 -- 5, we deal with the case of the {\it scalar metric}
$g(x)=\omega^2(x)\one$. The proof is based on an identity from
\cite{HR}. This identity shows that the values of $k_1$ such that
$\lambda(k_1 e_1+k_2 e_2)=\lambda$, are eigenvalues of a certain
non-selfadjoint operator $T_1(k_2,\lambda)$ (see Proposition
\ref{t1_prop} below). Our main observation is that the band edges
correspond to {\it degenerate} eigenvalues of that operator.
In Section 3, we introduce the operator $T_1$ 
and formulate the main technical result (Theorem \ref{main_tech}), which shows that the 
set of the values of $k_2$ for which $T_1(k_2,\lambda)$ may have degenerate eigenvalues, is discrete. Theorem \ref{main_tech} immediately implies the main result.
In Section 4, we
show that the condition of the operator $T_1(k_2,\lambda)$ having
degenerate eigenvalues is an analytic type condition. Hence,
either the set of ``degenerate'' $k_2$ is discrete, or the
operator $T_1(k_2,\lambda)$ has degenerate eigenvalues for all
$k_2\in \C$. In Section 5, we show that the latter case is
impossible for the free operator and hence, using perturbation
theory and estimates on the symbol, for the perturbed operator.
Section 6 describes the reduction of the case of a general $\c^2$-metric to the case of a
scalar one. In Section 7, we give an example of a {\it discrete}
periodic Schr\"odinger operator for which the statement of the
main theorem fails.
\section{The operator $T_1(k_2,\lambda)$}
In this section, we deal with the operator family 
\beq \label{300}
H(k)=(-i\nabla+\bark-A)^*\omega^2(-i\nabla+k-A)+V, 
\eeq 
which is a particular case of \eqref{hk_def}; here
$\omega$ is a scalar function satisfying \beq \label{301}
\omega\in \c^2_{\per}(\Omega), \qquad \omega^2 \ge m_g > 0. \eeq

Let $e_1$, $e_2$ be a standard basis in $\R^2$. We also denote the
coordinates of $k$ by $k_1,k_2$, that is, $k=k_1 e_1+k_2 e_2$, and
we will often denote $H(k)=H(k_1 e_1+k_2 e_2)$ by $H(k_1,k_2)$.
Since the statement of the main result is invariant under
rotations and dilations of $\R^2$, we can fix the following choice
of basis of the dual lattice: \beq \label{choiceofbasis}
b_1'=\alpha e_1,\quad b_2'=\beta e_1+e_2,\quad \text{where }
\,\alpha,\beta\in \R. \eeq

In the Hilbert space $\H^1_{\per}(\Omega)\oplus \L^2(\Omega)$, consider the following unbounded nonselfadjoint operator family:
\beq
\label{t1_def}
T_1(k_2,\lambda):=\begin{pmatrix}
0&\omega^{-2} I\\
-(H(0,k_2)-\lambda)\quad \quad \,\,&2\l(i\dd_1+A_1\r)-2i \omega^{-1}\dd_1\omega
\end{pmatrix},\quad k_2,\lambda\in\C,
\eeq
where $\dom(T_1(k_2,\lambda))=\H^2_{\per}(\Omega)\oplus \H^1_{\per}(\Omega)$, and $\dd_1=\frac{\dd}{\dd x_1}$.

The operator $T_1(k_2,\lambda)$ is introduced in order to ``linearize'' the equation $H(k_1,k_2)u=\lambda u$, considered as a quadratic eigenvalue problem in $k_1$, similarly to \cite[Lemma 3]{HR}. We summarize the properties of the family $T_1$ (most of which were also used in \cite{HR}) in the following proposition.
\begin{prop}
    \label{t1_prop}
    The operators $T_1(k_2,\lambda)$ satisfy the following properties.
    \begin{enumerate}
         \item[\rm (i)] For all $k_2,\lambda\in \C$, the operator $T_1(k_2,\lambda)$ is closed on the domain $\H^2_{\per}(\Omega)\oplus \H^1_{\per}(\Omega)$. As a consequence, the family $T_1(\cdot,\lambda)$ is an analytic type A operator family.
         \item[\rm (ii)] Suppose that $\lambda\notin \sigma(H(k_1,k_2))$. Then $k_1\notin \sigma(T_1(k_2,\lambda))$, and the resolvent 
         \beq
         \label{resolvent}
         \l(T_1(k_2,\lambda)-k_1\begin{pmatrix}
         I&0\\ 0&I
         \end{pmatrix}\r)^{-1}
         \eeq
          is compact in $\H^1_{\per}(\Omega)\oplus\L^2(\Omega)$.        
        \item[\rm (iii)] $k_1\in \sigma(T_1(k_2,\lambda))$ if and only if $\lambda\in \sigma(H(k))$, where $k=k_1 e_1+k_2 e_2$.
\item[\rm (iv)] For all $k_2,\lambda\in \C$, the set $\sigma(T_1(k_2,\lambda))$ is discrete in $\C$ and
         $2\pi\alpha$-periodic, where $\alpha$ is defined in
         \eqref{choiceofbasis}.
    \end{enumerate}
\end{prop}
\begin{proof}
{\it Part }(i). Clearly, the operator $T_1(k_2,\lambda)$ is bounded as an operator from $\H^2_{\per}(\Omega)\oplus \H^1_{\per}(\Omega)$ to $\H^1_{\per}(\Omega)\oplus \L^2(\Omega)$. We also have
$$
\l\|T_1(k_2,\lambda)\begin{pmatrix}
u\\v
\end{pmatrix}\r\|_{\H^1_{\per}(\Omega)\oplus \L^2(\Omega)}^2=
\|\omega^{-2}v\|_{H^1_{\per}(\Omega)}^2+
\|-(H(0,k_2)-\lambda)u+(2\l(i\dd_1+A_1\r)-2i \omega^{-1}\dd_1\omega)v\|^2_{L^2(\Omega)},
$$
from which it follows that the convergence in $T_1(k_2,\lambda)$-norm implies convergence of $v$ in $\H^1$ and convergence of $u$ in $\H^2$, so that $T_1(k_2,\lambda)$ is closed on its domain. Strong analyticity in $k_2$ and $\lambda$ follows directly from the definition.

{\it Part} (ii). Suppose that $\lambda\notin \sigma(H(k))$. Then the equation
$$
\l(T_1(k_2,\lambda)-k_1 \begin{pmatrix}
I&0\\0&I
\end{pmatrix}\r)\begin{pmatrix}
u\\v
\end{pmatrix}=\begin{pmatrix}
f\\g
\end{pmatrix}
$$
has a unique solution $\begin{pmatrix}
u\\v
\end{pmatrix}$ given by
\beq
\label{solution}
\begin{split}
u=(H(k)-\lambda)^{-1}\{(2i& \dd_1+2A_1 -2i\omega^{-1}\dd_1\omega-k_1) \omega^2 f-g\},\\
& v=\omega^2(f+k_1 u).
\end{split}
\eeq
Let $R(k,\lambda)=(H(k)-\lambda)^{-1}$. By plugging the expression for $u$ into the second equation of \eqref{solution}, we can rewrite \eqref{solution} in the operator form, applied to a vector $\begin{pmatrix}
f\\g
\end{pmatrix}\in \H^1_{\per}(\Omega)\oplus \L^2(\Omega)$:
\begin{multline*}
\l(T_1(k_2,\lambda)-k_1 \begin{pmatrix}
I&0\\0&I
\end{pmatrix}\r)^{-1}\begin{pmatrix}
f\\g
\end{pmatrix}
=\begin{pmatrix}
0&0\\\omega^2 I&0
\end{pmatrix}\begin{pmatrix}
f\\g
\end{pmatrix}+\\
+\begin{pmatrix}
I&0\\
0&k_1 \omega^2
\end{pmatrix}R(k,\lambda)\begin{pmatrix}
\l(2i\dd_1+2A_1(x)-k_1-2i\omega^{-1}(\dd_1\omega)\r)&\,\,\,-I\\
\l(2i\dd_1+2A_1(x)-k_1-2i\omega^{-1}(\dd_1\omega)\r)&\,\,\,-I
\end{pmatrix}
\begin{pmatrix}
\omega^2&0\\
0&I
\end{pmatrix}\begin{pmatrix}
f\\g
\end{pmatrix}.
\end{multline*}
The first operator in the right hand side is compact in $\H^1_{\per}(\Omega)\oplus \L^2(\Omega)$, because the embedding $\H^1_{\per}(\Omega)\subset \L^2(\Omega)$ is compact. The operator in the second term is compact since $R(k,\lambda)$ is bounded as an operator from $\L^2(\Omega)$ to $\H^2_{\per}(\Omega)$ and hence is compact from $\L^2(\Omega)$ to $\H^1_{\per}(\Omega)$. Hence, {\it the resolvent of $T_1(k_2,\lambda)$ is compact}, which completes the proof of Part (ii). 

{\it Part} (iii). The ``only if'' part is included in Part (ii). To establish ``if'' part, suppose that $H(k)u=\lambda u$. Then $u\in \H^2_{\per}(\Omega)$, and
$$
T_1(k_2,\lambda)\begin{pmatrix}
u\\ k_1 \omega^2 u
\end{pmatrix}=k_1 \begin{pmatrix}
u\\ k_1 \omega^2  u
\end{pmatrix}.
$$
This completes the proof of (iii).

{\it Part} (iv). From the proofs of the absolute continuity of the spectrum (see e.g., \cite{BSu}), it follows that, for any $\lambda, k_2\in \C$, the set $\{k_1\colon \lambda\in \sigma(H(k_1,k_2))\}$ is discrete. Hence, there exists at least one value of $k_1$ such that the resolvent \eqref{resolvent} exists, which, together with Part (ii), implies that $\sigma(T(k_2,\lambda))$ is discrete.
Periodicity of the spectrum follows from the fact that $H(k)$ is unitarily equivalent to $H(k+b')$ for any $b'\in \Gamma'$ (see \eqref{lastp2}), and so $H(k_1,k_2)$ is unitarily equivalent to $H\l(k_1+2\pi\alpha,k_2\r)$.
\end{proof}

In the sequel, by ``the multiplicity of an isolated eigenvalue''
we will mean {\it algebraic multiplicity}, \ie the dimension of
the range of the corresponding Riesz projection. We will call an
eigenvalue {\it degenerate} if its algebraic multiplicity is
greater than or equal to 2. Otherwise, an eigenvalue is called
{\it simple}.
\begin{lemma}
    \label{degenerate}
    Suppose that a band function $\lambda_j(\cdot)$ attains its local minimum or maximum value
    $\lambda_*$ at $k^*=k_1^* e_1+k_2^* e_2\in \R^2$. Then $k_1^*$ is an eigenvalue of $T_1(k_2^*,\lambda_*)$ of the $($algebraic$)$ multiplicity at least two.
\end{lemma}
\begin{proof}
    By Proposition \ref{t1_prop} $k_1^*$ is an eigenvalue of
    $T_1(k_2^*,\lambda_*)$.
    For some $\er>0$, there are no other eigenvalues of $T_1(k_2^*,\lambda_*)$ within the closed disc $\overline {B_{\er}(k_1^*)}$.
    Let
    $$
    P(k_2^*,\lambda):=-\frac{1}{2\pi i}\oint_{\partial B_{\er}(k_1^*)} (T_1(k_2,\lambda)-\ka I)^{-1}d\ka
    $$
    be the Riesz projection. The standard arguments \cite[Section IV.3.5]{Kato} show that, for some $\delta>0$,
    $\rank P(k_2^*,\lambda)$ is continuous in $\lambda$ as long as $|\lambda-\lambda_*|<\delta$.
    Without loss of generality, assume that $k^*$ is a local minimum of $\lambda_j(\cdot)$.

    From the proofs of the absolute continuity of the spectrum (see e.g., \cite{BSu}), it follows that $\lambda_j(k)$ cannot be constant in $k_1$ on any interval.
    Then, for a sufficiently small $\delta>0$, the equation $\lambda_j(k_1,k_2^*)=\lambda_*+\delta$ has at least
    two different solutions as an equation in $k_1$ (note that these arguments do not use any analyticity of $\lambda_j(\cdot,k_2^*)$, only continuity). Hence, by Part (iii) of 
    Proposition \ref{t1_prop}, $\rank P(k_2^*,\lambda_*+\delta)\ge 2$ for all sufficiently small $\delta$,
    and therefore $\rank P(k_2^*,\lambda_*)\ge 2$ due to continuity.
\end{proof}
\noindent The following is the main technical result of the paper.
\begin{theorem}
    \label{main_tech}
    Suppose that the coefficients $\omega,A,V$ satisfy \eqref{301},\eqref{a_cond},\eqref{v_cond}. For any $\lambda\in \R$, the set
    $$
        \{k_2\in \R\colon \text{the operator }T_1(k_2,\lambda)\text{ has at least one real degenerate eigenvalue}\}
    $$
    is discrete.
\end{theorem}
{\noindent \bf Proof of Theorem $\ref{main}$: the case of a scalar
metric. }Fix a band function $\lambda_j(\cdot)$ and assume that
$\lambda_*$ is a minimum or a maximum of $\lambda_j$. From Theorem
\ref{main_tech} and Lemma \ref{degenerate}, the set of possible
$k_2$ such that for some $k_1$ we have $\lambda_j(k)=\lambda_*$,
is discrete. For each of these $k_2$, the set of possible values
of $k_1$ is also discrete by Proposition \ref{t1_prop}.\,\qed

Theorem \ref{main_tech} is proved in Sections \ref{proof331} and \ref{proof332}. The rest of the proof of Theorem \ref{main} is a (mostly standard) argument
of transforming a general metric to a scalar metric by introducing isothermal coordinates. This is done in Section \ref{variablemetric}. 

\section{Proof of Theorem \ref{main_tech}}
\label{proof331}
Let
$$
p(z)=z^n+a_{n-1}z^{n-1}+\ldots+a_0
$$
be a monic polynomial with roots $z_1,\ldots,z_n$. The {\it discriminant} of $p$ is defined as
$$
\Delta(p)=\prod_{1\le i<j\le n}(z_i-z_j)^2.
$$
It is clear that $\Delta(p)$ vanishes if and only if $p$ has roots
of multiplicity greater than or equal to 2. It is well known (see,
for example, \cite[Section 5.9]{WW}) that $\Delta(p)$ is a
polynomial function of the coefficients $a_0,\ldots,a_{n-1}$. 

The proof of the following lemma can be extracted from a slightly different and more abstract setting of \cite{Ku_book,Kuch_Z}. For the convenience of the reader, we include the argument.
\begin{lemma}
    \label{discriminant}
    Let $\mathcal C$ be a simple closed piecewise smooth contour in $\C$, and let $\{T(z),z\in \mathcal D\}$ be
    an operator family of type A in a Hilbert space ${\mathcal H}$ analytic in a simply connected domain $\mathcal D\subset \C$.
	Suppose that, for all $z\in \mathcal D$, the spectrum of $T(z)$ in the interior of $\mathcal C$ is discrete and finite,
    and $\sigma(T(z))\cap \mathcal C=\varnothing$. Then the set
    $$
    \{z\in \mathcal D\colon T(z)\text { has at least one degenerate eigenvalue in the interior of }\mathcal C\}
    $$
    is a null-set of a function analytic in $\mathcal D$, and hence this set either coincides with $\mathcal D$ or is discrete in $\mathcal D$.
\end{lemma}
\begin{proof}
    Let
    $$
    P(z):=-\frac{1}{2\pi i}\oint_{\mathcal C} (T(z)-\ka I)^{-1}d\ka
    $$
    be the Riesz projection. By assumption, $n:=\rank
    P(z)=\mathrm{const}$ is finite and independent on $z$, and $P(z)$ is analytic in $\mathcal D$.
    Fix $z_0\in\mathcal D$. 
    The results of \cite[Section VII.1.3]{Kato} imply that there exists an analytic in $\mathcal D$ bounded operator-valued function $U\colon \mathcal{D}\to \mathcal{B(H)}$\footnote{$\mathcal{B(H)}$ denotes the algebra of bounded operators on a Hilbert space $H$.} such that
    $U(\cdot)^{-1}$ is also analytic in $\mathcal D$ and    $P(z)=U(z)P(z_0)U(z)^{-1}$. Take
    $$
        T_0(z):=\l.U(z)^{-1} T(z)U(z)\r|_{\ran P(z_0)}.
    $$
    The family $T_0(z)$ is an analytic operator family acting in a fixed
    finite-dimensional space that has the same eigenvalues and multiplicities as $T(z)$
    restricted to $\ran P(z)$. The monic polynomial $p_z(\ka)=(-1)^n\det(T_0(z)-\ka)$
    is the characteristic polynomial of $T_0(z)$ and
    has the coefficients analytic in $\mathcal D$ (in the variable $z$).
    Hence, its discriminant $\Delta(p_z)$ is also an analytic function in $\mathcal D$ vanishing if and only if $T_0(z)$
    (and, as a consequence, $T(z)$) has degenerate eigenvalues in the interior of $\mathcal C$.
\end{proof}
Recall that we had a special choice of basis in $\Gamma'$,
$$
b_1'=\alpha e_1,\quad b_2'=\beta e_1+e_2.
$$
Let also
$$
k=k_1 e_1+k_2 e_2, \quad k_1=r_1+i l_1,\quad k_2=r_2+i l_2.
$$
The following two theorems are the main technical statements of the paper. We postpone the proofs to the next section.

Without loss of generality, one can assume $\lambda=0$ by choosing
a different $V$. In the sequel, we will make this assumption and
{\it drop $\lambda$ from the notation for $T_1$}, that is,
$T_1(k_2):=T_1(k_2,0)$.
\begin{theorem}
    \label{tech1}
    Let $\delta>0$.
    There exist $C=C(A,V,\omega)$ and $C_1=C_1(A,V,\omega,\delta)\in 2\pi\Z$ such that operator $H(k)$ defined in \eqref{300} is invertible and satisfies
    $$
        \|H(k)^{-1}\|\le \frac{C}{|l_1| \delta^2}
    $$
    provided that $\dist(r_2,2\pi \Z)\ge \delta$, $l_1\in 2\pi \Z$, $|l_1|\ge C_1$.
    As a consequence, the horizontal lines $\im k_1=\pm C_1$ have empty intersection with
    $\sigma(T_1(k_2))$.
\end{theorem}

\begin{theorem}
    \label{tech2}
    There exists $l=l(A,V,\omega)\in 2\pi\Z$ such that, for all $n\in 2\pi\Z$, the spectrum of $T_1(k_2)$ is simple for $k_2=\frac{\pi}{2}+n+i\l(\frac{\pi}{2}+l\r)\alpha$.
\end{theorem}
\vskip 2mm \noindent {\bf Proof of Theorem \ref{main_tech}.}
Assume the contrary, \ie that the set of $k_2\in \R$ for which
$T_1(k_2)$ has real degenerate eigenvalues has a limit point
$k_2^{(0)}$ (as above we assume $\lambda=0$). Let us consider two
cases. \vskip 1mm \noindent {\it Case $1$.} Suppose that
$\dist(k_2^{(0)},2\pi \Z)>0$. Take
$\delta=\min\{\pi/2,\dist(k_2^{(0)},2\pi \Z)\}$. There exists a
single $n\in 2\pi \Z$ such that $k_2^{(0)}\in
[n+\delta,n+2\pi-\delta]$. Let $\mathcal C_0$ be a path in the
$k_2$-plane starting at $k_2^{(0)}$, then going straight towards
the point $\frac{\pi}{2}+n$, and then going vertically towards the
point $k_2^{(1)}:=\frac{\pi}{2}+n+i\l(\frac{\pi}{2}+l\r)\alpha$
from Theorem \ref{tech2}.

The points
$k_2\in \mathcal C_0$ satisfy the assumptions of Theorem \ref{tech1}.
Let us consider the eigenvalues of $T_1(k_2)$ lying within the strip $|\im k_1|< C_1$, where $C_1$ is the constant from Theorem \ref{tech1}. They form a discrete $2\pi\alpha$-periodic set. For each $k_2\in \mathcal C_0$ there exists a point $r(k_2)\in \R$ which is {\it not} a real part of any of these eigenvalues. Moreover, by continuity arguments, this also holds in a small (complex) neighbourhood of $k_2$. Let us cover $\mathcal C_0$ by a finite number of these neighbourhoods $\mathcal D_j$, $j=1,\ldots,p$, so that $k_2^{(0)}\in \mathcal D_1$ and
$k_2^{(1)}\in \mathcal D_p$
and denote the corresponding values of $r(k_2)$ by $r_j$.
For each $j$, denote by $\mathcal C_j$ the boundary of the following rectangle:
$$
r_j<\re k_1<r_j+2\pi\alpha,\quad -C_1< \im k_1< C_1.
$$
Informally speaking, each rectangle contains all eigenvalues that we are interested in: they are initially on the real line, they cannot cross the lines $\im k_1=\pm C_1$, and the pictures to the right and to the left copy the picture in the rectangle due to periodicity.

Let us apply Lemma \ref{discriminant} to each of the domains $\mathcal D_j$ and contours $\mathcal C_j$. Due to Theorem \ref{tech2}, the
spectrum of $T_1(k_2^{(1)})$ is simple, and hence the set of ``degenerate'' $k_2$ should be discrete in a neighbourhood $\mathcal D_p$ of $k_2^{(1)}$. By the standard arguments of analytic continuation, it should also be
discrete in every neighbourhood $\mathcal D_1\ldots,\mathcal D_p$. However, since $k_2^{(0)}\in \mathcal D_1$, it is not discrete in $\mathcal D_1$, which is a contradiction.

\vskip 1mm
\noindent {\it Case $2$.} Suppose that $k_2^{(0)}\in 2\pi \Z$. The set of {\it real} eigenvalues of $T_1(k_2^{(0)})$ is, again, discrete and $2\pi \alpha$-periodic. Let us surround the eigenvalues on one period by a contour $\mathcal C$ not containing any other eigenvalues. In a small neighbourhood $\mathcal D_0$ of $k_2^{(0)}$, these eigenvalues still stay within $\mathcal C$. Apply Lemma \ref{discriminant} to $\mathcal C$ and $\mathcal D_0$. Again, since the set of ``degenerate'' values of $k_2$ is not discrete in $\mathcal D_0$, it should coincide with $\mathcal D_0$, and hence there exists at least one more point with the same property that belongs to $\R\setminus 2\pi \Z$, and thus the situation reduces to Case 1.\,\qed
\section{Proofs of Theorems \ref{tech1}, \ref{tech2}}
\label{proof332}
Let us start by recalling some notation introduced above,
$$
b_1'=\alpha e_1,\quad b_2'=\beta e_1+e_2,\quad \alpha,\beta\in \R;
$$
$$
k=k_1 e_1+k_2 e_2, \quad k_1=r_1+i l_1,\quad k_2=r_2+i l_2,\quad r_1,r_2,l_1,l_2\in \R.
$$
In this section, we will emphasize the dependence of $H$ on $g,A,V$ and use the notation $H(k;g,A,V)$. Consider the free operator $H_0(k):=H(k;\one,0,0)$. Its eigenfunctions are of the form
$$
\exp\{i m\cdot x\}=\exp\{i(m_1 b_1'+m_2 b_2')\cdot(x_1 e_1+x_2 e_2)\}=\exp\{i((\alpha m_1+\beta m_2)x_1+m_2 x_2)\},
$$
$$
m = m_1 b_1' + m_2 b_2' \in \Gamma',\quad  m_1,m_2\in 2\pi \Z,
$$
and
$$
H_0(k)\exp\{i m\cdot x\}=((-i \dd_1+k_1)^2+(-i \dd_2+k_2)^2)\exp\{i m\cdot x\}=h_{m}(k)\exp\{i m\cdot x\},
$$
where $h_m(k)$ is the symbol of $H_0(k)$:
$$
h_m(k)=(\alpha m_1+\beta m_2+k_1)^2+(m_2+k_2)^2=q^+_{m}(k)q^-_m(k),
$$
$$
q^{\pm}_m(k)=\alpha m_1+\beta m_2+r_1\mp l_2+i(l_1\pm m_2\pm r_2).
$$
Let also $Q^{\pm}(k)$ be the operators with symbols $q^{\pm}_m(k)$ respectively, so that $H_0(k)=Q^+(k)Q^-(k)$.
Suppose that the magnetic potential $A$ satisfies \eqref{a_cond}. Then there exists a $\Gamma$-periodic scalar function $\varphi\in C_{\per}^2(\Omega)$ such that
\beq
\label{phi_prop}
(\nabla\varphi)(x)=A_2(x)e_1-A_1(x)e_2,\quad \int_\Omega\varphi(x)\,dx=0,\quad \|\varphi\|_{\c^2(\Omega)}\le C \|A\|_{\c^1(\Omega)}.
\eeq
Let also
$$
B(x)=\dd_1 A_2(x)-\dd_2 A_1(x), \quad w(x):=e^{-2\varphi(x)}.
$$
The operator $H(k;\one,A,B)$ is called {\it the Pauli operator} (more precisely, a block of the Pauli operator).
The following is proved in \cite{BSu} and allows us to reduce the case of the magnetic
potential, essentially, to the case of the free operator.
\begin{prop}
\label{commutator}
Under the above assumptions, if $Q^{\pm}(k)$ are invertible, then $H(k;\one,A,B)$ is also invertible, and
\begin{multline}
\label{commutator_f}
H(k;\one,A,B)^{-1}=e^{\varphi}Q^-(k)^{-1}e^{-2\varphi}Q^+(k)^{-1}e^{\varphi}=\\
=e^{\varphi(x)}H_0(k)^{-1} \l\{e^{-\varphi}+(-i\dd_1w+\dd_2w)Q^+(k)^{-1}e^{\varphi}\r\}.
\end{multline}
\end{prop}
\noindent The following proposition can also be easily verified, see \cite{BSu99}.
It will be used to reduce the case of a scalar metric $g=\omega^2\one$ to the case $g=\one$.
\begin{prop}
Suppose that $\omega\in \c_{\per}^2(\Omega)$, $V\in \L^{\infty}(\Omega)$, $A\in \c_{\per}^1(\Omega)$. Then
\label{scalarmetric}
\beq
\label{521}
H(k;\omega^2\one,A,V)=\omega H(k;\one,A,\omega^{-2}V+\omega^{-1}\Delta \omega)\omega,
\eeq
\beq
\label{522}
\omega H(k;\one,A,V)\omega=H(k;\omega^2\one,A,\omega^2 V-\omega\Delta\omega).
\eeq
\end{prop}
\vskip 2mm
\noindent {\bf Proof of Theorem \ref{tech1}.} Suppose that $\dist(r_2,2\pi \Z)=\delta$. Since $l_1\pm m_2\in 2\pi \Z$, we have $|q^{\pm}_m(k)|\ge \delta$. In addition, $\im q_m^+(k)+\im q_m^-(k)=2 l_1$, and hence we either have $|q^+_m(k)|\ge |l_1|$ or $|q^-_m(k)|\ge |l_1|$. Combining these estimates, we obtain $|h_m (k)| \ge |l_1| \delta$, and
\beq
\label{h0q}
\|H_0(k)^{-1}\|\le \frac{1}{|l_1|\delta},\quad \|Q^+(k)^{-1}\|\le \frac{1}{\delta},
\eeq
which completes the proof for $A=0$, $V=0$, $\omega=1$. If $A\neq 0$
and $V(x)=B(x)$, then, from \eqref{commutator_f} and \eqref{h0q}, we get
\beq
\label{hkbound}
\|H(k;\one,A,B)^{-1}\|\le \frac{C}{|l_1|\delta^2},
\eeq
where $C$ depends on $A$ via $w$ and $\varphi$. The standard Neumann series
arguments imply that the bound \eqref{hkbound} holds for the operator $H(k,\one,A,V)$ with  arbitrary $V\in \L^{\infty}(\Omega)$ (and maybe a different $C$) for sufficiently large $l_1$, say,
$$
|l_1|\ge \frac{2\|V-B\|_{\L^{\infty}(\Omega)}C}{\delta^2}.
$$
The case of arbitrary $\omega$ follows from Proposition \ref{scalarmetric}.
\,\qed
\vskip 2mm

We now make some preparations for the proof of Theorem \ref{tech2}. Fix $k_2$ as in the formulation of the theorem, so that
\beq
\label{choice_k2}
r_2=\frac{\pi}{2}+n,\quad l_2=\l(\frac{\pi}{2}+l\r)\alpha, \quad  l,n\in 2\pi  \Z.
\eeq
For these $k_2$, define
$$
\Sigma_n:=\{k_1\in \C\colon h_m(k_1,k_2)=0\,\text{ for some }m_1,m_2\in 2\pi \Z\}.
$$
In other words, it is the set of $k_1$ for which $H_0(k_1,k_2)$ is not invertible. A simple computation shows that
$\Sigma_n$ consists of points $r_1+i l_1$ of the following form:
\beq
\label{sigma_desc}
\begin{cases}
    r_1=-\alpha m_1-\beta m_2 \mp \l(\frac{\pi}{2}+l\r)\alpha\\
    l_1=\pm \l(\frac{\pi}{2}+n+ m_2\r)
\end{cases},\quad m_1,m_2\in 2\pi\Z.
\eeq
Since one can replace the variables $m_1$ by $m_1+l$, one can see that {\it the set $\Sigma_n$ does not depend on $l$}.

Let us describe the set $\Sigma_n$ in more detail. First of all, it is easy to see that different values of $(m_1,m_2)$ give different points of $\Sigma_n$,
as $m_2$ and the signs are uniquely determined by the value of $l_1$,
and $m_1$ is determined by $r_1$ afterwards. Next, the set $\Sigma_n$ lies on the union of horizontal lines $\im k_1\in \pi/2+\pi \Z$. On each line, it is a sequence of equally spaced points with the
spacings $2\pi \alpha$.

We will also need another set $G_n$ defined by
$$
G_n:=(\R+i \pi \Z)\cup \bigcup_{z\in \Sigma_n}\l(z+\pi\alpha+i\l[-\frac{\pi}{2},\frac{\pi}{2}\r]\r).
$$
The set $G_n$ consists of horizontal lines $\im k_1 \in \pi \Z$ separating the horizontal lines of $\Sigma_n$. In addition, for each point of $\Sigma_n$, we include a vertical line segment of the length $\pi$ separating this point from the next point of $\Sigma_n$ lying on the same line. One can imagine $G_n$ as a ``brick wall'' consisting of rectangles such that there is exactly one element of $\Sigma_n$ inside of each rectangle.

\begin{figure}[H]
\begin{tikzpicture}[scale=0.8]
\draw [->](0,0) -- (16,0);
\draw [->](0,-7) -- (0,7);
\node[above left] at (0,0){$O$};
\node[left]at (0,7){$\im k_1$};
\node at (0,0) {$\times$};
\node [below] at (16,0){$\re k_1$};

\draw [<->] (2.2,-6.5)--(5.2,-6.5);
\node [below]at (3.7,-6.5) {$2\pi\alpha$};
\draw [dashed](-1,1)--(2,1);
\node [left] at (-1,1) {$\pi/2$};
\draw [dashed](4.3,-1.5) -- (4.3,3);
\draw [dashed](4,-1.5) -- (4,-1);
\draw [<->] (4.3,-1.5)--(4,-1.5);
\node [below] at (4.15,-1.4) {$2\pi\beta$};

\node [left] at (-1,0) {$0$};
\node [left] at (-1,2) {$\pi$};
\node [left] at (-1,4) {$2\pi$};
\node [left] at (-1,6) {$3\pi$};
\node [left] at (-1,-2) {$-\pi$};
\node [left] at (-1,-4) {$-2\pi$};
\node [left] at (-1,-6) {$-3\pi$};

\draw[fill] (2.3,-3) circle [radius=0.1];
\draw[fill] (5.3,-3) circle [radius=0.1];
\draw[fill] (8.3,-3) circle [radius=0.1];
\draw[fill] (11.3,-3) circle [radius=0.1];
\draw[fill] (14.3,-3) circle [radius=0.1];

\draw[fill] (2,1) circle [radius=0.1];
\draw[fill] (5,1) circle [radius=0.1];
\draw[fill] (8,1) circle [radius=0.1];
\draw[fill] (11,1) circle [radius=0.1];
\draw[fill] (14,1) circle [radius=0.1];

\draw[fill] (1.7,5) circle [radius=0.1];
\draw[fill] (4.7,5) circle [radius=0.1];
\draw[fill] (7.7,5) circle [radius=0.1];
\draw[fill] (10.7,5) circle [radius=0.1];
\draw[fill] (13.7,5) circle [radius=0.1];

\draw (1.3,3) circle [radius=0.1];
\draw (4.3,3) circle [radius=0.1];
\draw (7.3,3) circle [radius=0.1];
\draw (10.3,3) circle [radius=0.1];
\draw (13.3,3) circle [radius=0.1];

\draw (1,-1) circle [radius=0.1];
\draw (4,-1) circle [radius=0.1];
\draw (7,-1) circle [radius=0.1];
\draw (10,-1) circle [radius=0.1];
\draw (13,-1) circle [radius=0.1];

\draw (0.7,-5) circle [radius=0.1];
\draw (3.7,-5) circle [radius=0.1];
\draw (6.7,-5) circle [radius=0.1];
\draw (9.7,-5) circle [radius=0.1];
\draw (12.7,-5) circle [radius=0.1];
\draw [ultra thick] (-1,2)--(15,2);
\draw [ultra thick] (-1,0)--(15,0);
\draw [ultra thick] (-1,-2)--(15,-2);
\draw [ultra thick] (-1,-4)--(15,-4);
\draw [ultra thick] (-1,4)--(15,4);
\draw [ultra thick] (-1,6)--(15,6);
\draw [ultra thick] (-1,-6)--(15,-6);
\draw [ultra thick] (0.2,4)--(0.2,6);
\draw [ultra thick] (3.2,4)--(3.2,6);
\draw [ultra thick] (6.2,4)--(6.2,6);
\draw [ultra thick] (9.2,4)--(9.2,6);
\draw [ultra thick] (12.2,4)--(12.2,6);

\draw [ultra thick] (0.5,0)--(0.5,2);
\draw [ultra thick] (3.5,0)--(3.5,2);
\draw [ultra thick] (6.5,0)--(6.5,2);
\draw [ultra thick] (9.5,0)--(9.5,2);
\draw [ultra thick] (12.5,0)--(12.5,2);

\draw [ultra thick] (0.8,-4)--(0.8,-2);
\draw [ultra thick] (3.8,-4)--(3.8,-2);
\draw [ultra thick] (6.8,-4)--(6.8,-2);
\draw [ultra thick] (9.8,-4)--(9.8,-2);
\draw [ultra thick] (12.8,-4)--(12.8,-2);
\draw [ultra thick] (-0.2,4)--(-0.2,2);
\draw [ultra thick] (2.8,4)--(2.8,2);
\draw [ultra thick] (5.8,4)--(5.8,2);
\draw [ultra thick] (8.8,4)--(8.8,2);
\draw [ultra thick] (11.8,4)--(11.8,2);
\draw [ultra thick] (14.8,4)--(14.8,2);

\draw [ultra thick] (-0.5,0)--(-0.5,-2);
\draw [ultra thick] (2.5,0)--(2.5,-2);
\draw [ultra thick] (5.5,0)--(5.5,-2);
\draw [ultra thick] (8.5,0)--(8.5,-2);
\draw [ultra thick] (11.5,0)--(11.5,-2);
\draw [ultra thick] (14.5,0)--(14.5,-2);

\draw [ultra thick] (-0.8,-4)--(-0.8,-6);
\draw [ultra thick] (2.2,-4)--(2.2,-6);
\draw [ultra thick] (5.2,-4)--(5.2,-6);
\draw [ultra thick] (8.2,-4)--(8.2,-6);
\draw [ultra thick] (11.2,-4)--(11.2,-6);
\draw [ultra thick] (14.2,-4)--(14.2,-6);
\end{tikzpicture}
\caption{The sets $\Sigma_n$ and $G_n$.}
\label{sigmafig}
\end{figure}
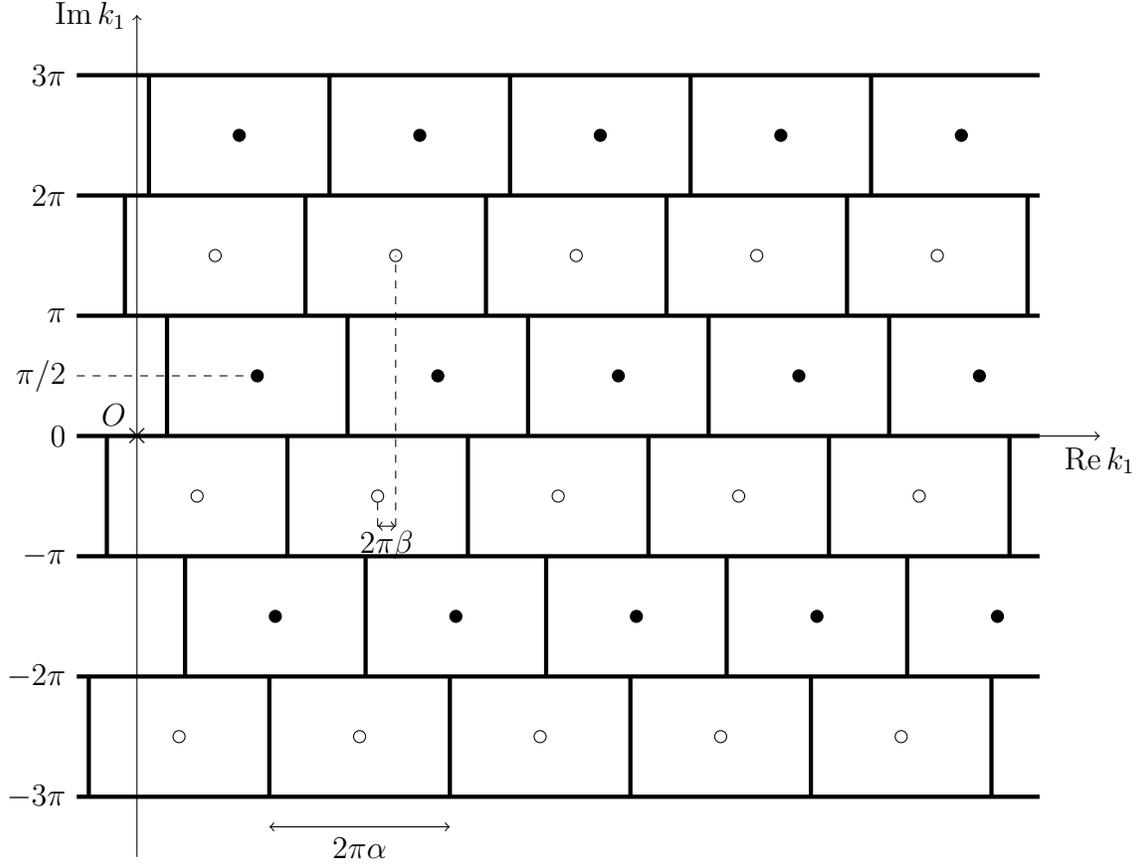
On Figure \ref{sigmafig}, an example of $\Sigma_n$ and $G_n$ is shown for $n=0$, $\alpha=0.75$, $\beta=0.075$. The set $G_n$ is represented by thick lines, and the locations of points of $\Sigma_n$ are indicated by black and white circles,
corresponding to the upper or lower choice of signs in \eqref{sigma_desc}, respectively.
\begin{lemma}
    \label{tech2_aux}
    Suppose that $k_1\in G_n$, $k_2=\frac{\pi}{2}+n+i\l(\frac{\pi}{2}+l\r)\alpha$, where $l,n\in 2\pi \Z$. Then
    $$
        |h_m(k)|\ge C|l|
    $$
    uniformly in $m_1,m_2\in 2\pi \Z$.
\end{lemma}
\begin{proof}
    Since $|\re q_m^+(k)-\re q_m^-(k)|=2|l_2|\ge C|l|$, we have for each $m$ either $|q_m^+(k)|\ge \frac12 C|l|$ or $|q_m^-(k)|\ge \frac12 C|l|$. On the other hand, both $|q_m^+(k)|$ and $|q_m^-(k)|$ are distances from $k_1$ to a certain point of $\Sigma_n$, which is bounded from below by a positive constant,
\beq
\label{qlower}
        |q_m^{\pm}(k)|\ge \dist(k_1,\Sigma_n)\ge \dist(G_n,\Sigma_n)=\min\l\{\frac{\pi}{2},\pi\alpha\r\}.
\eeq
    The combination of these estimates completes the proof of the lemma.
\end{proof}
\begin{remark}
Lemma \ref{tech2_aux} is the main ingredient of the proof that relies on the assumption $d=2$. In $d\ge 3$, one cannot construct a set $G_n$ with similar properties and constant size of the bricks.
\end{remark}
\begin{cor}
\label{bricks}
Under the assumptions of Lemma $\ref{tech2_aux}$, there exists $L_0(A,V,\omega)>0$ such that, if $|l|>L_0(\omega,A,V)$, then the operator $H(k;\omega^2\one,A,V)$ is invertible and 
$$
\|H(k;\omega^2\one,A,V)^{-1}\| \le \frac{C(\omega,A,V)}{|l|},
$$
where the constants $C$ and $L_0$ depend only on $\|A\|_{\c^1(\Omega)}$, $\|V\|_{\L^{\infty}(\Omega)}$, $\|\omega\|_{\c^2(\Omega)}$ and on the constant $m_g$ from $\eqref{301}$.
\end{cor}
\begin{proof}
From Proposition \ref{scalarmetric}, we have
$$
\|H(k;\omega^2\one,A,V)^{-1}\|\le m_g^{-2}\|H(k;\one,A,V_{\omega})^{-1}\|,
$$
where
$$
V_\omega = \omega^{-2} V + \omega^{-1} \Delta\omega
$$
and therefore
$$
\|V_\omega\|_{L^\infty(\Omega)} \le
m_g^{-2} \|V\|_{L^\infty(\Omega)} +
m_g^{-1} \|\omega\|_{C^2(\Omega)}.
$$
From \eqref{commutator_f},\eqref{phi_prop}, Lemma \ref{tech2_aux}, and \eqref{qlower}, we have
$$
\|H(k;\one,A,B)^{-1}\|\le C(A)\|H_0(k)^{-1}\|\le
\frac{C_1(A)}{|l|},
$$
where $C(A)$, $C_1(A)$ depend only on $\|A\|_{\c^1(\Omega)}$. Since $\|B\|_{\L^{\infty}(\Omega)}\le 2\|A\|_{\c^1(\Omega)}$, we can use the same Neumann series argument as in the proof of Theorem \ref{tech1} to replace $B$ by $V_{\omega}$.
\end{proof}
\vskip 2mm
\noindent {\bf Proof of Theorem \ref{tech2}.} Denote by $T_{\mu}(k_2)$ the operator $T_1(k_2)$ with $V$, $A$, $\omega$ replaced by $\mu V$, $\mu A$ and $\mu\omega+(1-\mu)$ respectively. It is a one-parametric family connecting the ``free'' operator $T_0(k_2)$ with $T_1(k_2)$.

It is easy to see that $\sigma(T_0(k_2))=\Sigma_n$, because $\Sigma_n$ is exactly the set
of $k_1\in \C$ for which the symbol of $H_0(k)$ is not invertible. Moreover, an easy computation shows that, for each $k_1\in \Sigma_n$, the corresponding eigenspace is one-dimensional and is spanned by $\begin{pmatrix}
e^{im\cdot x}\\k_1 e^{im\cdot x}
\end{pmatrix}$,
where $m$ is determined by $k_1$ via \eqref{sigma_desc}. Note that each value of $m$ appears twice (for two different values of $k_1$)
because of two possible signs. Hence, the total collection of eigenvectors spans $\H^1_{\per}(\Omega)\oplus \L^2(\Omega)$, so there are no Jordan cells and {\it the spectrum of $T_0(k_2)$ is simple}.

It remains to prove that $T_1(k_2)$ also has simple spectrum. Consider the Riesz projection of $T_{\mu}(k_2)$ with respect to the boundary of some rectangle of $G_n$. For $\mu=0$, the rectangle contains exactly one simple eigenvalue, and the range of the projection has dimension 1. Let us increase $\mu$. The only way for the dimension of the range to change is to have an eigenvalue of $T_{\mu}(k_2)$ approach the set $G_n$. This, however, is impossible for $\mu\in [0,1]$ due to Corollary \ref{bricks}, and hence the eigenvalues of $T_1(k_2)$ stay simple.\,\qed
\begin{remark}
The proof of Theorem \ref{tech2} is based on the ideas of \cite[Section VI]{HR}.
\end{remark}
\section{The case of variable metric}
\label{variablemetric}
In this section we show how to reduce the case of an operator with arbitrary metric $g$ satisfying \eqref{g_cond} to the case of the scalar metric. The technical difference with
standard arguments such as in \cite{Sht1} is that we need to keep track of the quasimomentum in order to ensure that it is transformed linearly. This is done by an
additional ``gauge transformation''. The following proposition establishes the existence of {\it global isometric coordinates} in which the metric $g$ becomes scalar. See \cite[Proposition 18]{KuL} for the proof.
\begin{prop}
    \label{conformal}
    Suppose that $g$ satisfies $\eqref{g_cond}$. Then there exists a basis $b_1^*,b_2^*$ of $\R^2$ and a one-to-one map $\Psi\colon \R^2\to \R^2$, $\Psi\in \c^3(\R^2)$, $\det\Psi'(x)\neq 0$,
    $$
    \Psi(0)=0,\quad \Psi(x+n_1 b_1+n_2 b_2)=\Psi(x)+n_1 b_1^*+n_2 b_2^*,\quad\forall n_1, n_2\in \Z,
$$
    such that
\beq
\label{g_relation}
        |\det\Psi'(x)|^{-1}\Psi'(x)g(x)\Psi'(x)^t=\omega^2(\Psi(x))\one,
\eeq
where $\omega\in C^2(\R^2)$ is a strictly positive scalar function periodic with respect to the lattice $\Gamma_*$ spanned by $b_1^*$ and $b_2^*$.
\end{prop}
Let us introduce  some notation. Suppose that the operator $H(g,A,V)$ satisfies
the assumptions of Theorem \ref{main}. Let $\Psi$ be the transformation obtained from Proposition \ref{conformal}. Denote by $T_*\colon \R^2\to \R^2$ the linear transformation defined by $T_*(b_1)=b_1^*$, $T_*(b_2)=b_2^*$. The transformation $T_*$, as well as the map $\Psi$, transforms the lattice $\Gamma$ into $\Gamma_*$. Let also
$$
y=\Psi(x),\quad  A_*(y)=(\Psi'(x)^{-1})^{t}A(x),\quad  V_*(y)=\psi_*(y)^{-2}V(x),\quad \psi_*(y)=|\det \Psi'(x)|^{1/2}.
$$
Let also
$$
\Omega_{\Psi}=\Psi(\Omega),\quad \Omega_*=\{y_1 b_1^*+y_2 b_2^*,\quad y_1,y_2\in [0,1)\}.
$$
Note that both $\Omega_*$ and $\Omega_{\Psi}$ are fundamental domains of $\Gamma_*$, and there is a natural correspondence between $L^2(\Omega_*)$ and $L^2(\Omega_{\Psi})$, as both can be identified with $\R^d/\Gamma_*$.
\begin{lemma}
\label{changevariables}
In the above notation, let $\Phi\colon \L^2(\Omega)\to \L^2(\Omega_*)$ be the unitary operator of change of variables:
$$
u(x)=\psi_*(y)(\Phi u)(y),\quad y=\Psi(x),
$$
where $u$ is considered as an element of $L^2(\Omega_*)$. Then
$$
\Phi H(0;g,A,V)\Phi^{-1}=\psi_* H(0;\omega^2\one,A_*,V_*)\psi_*.
$$
\end{lemma}
\begin{proof}
Let $v=\Phi u$, and let us extend it $\Gamma_*$-periodically into $\R^d$. Then, due to \eqref{g_relation} and the change of variable rule, the quadratic form of the left hand side applied to $v$ is equal to
$$
(H(0;g,A,V)\Phi^{-1}v,\Phi^{-1}v)_{\L^2(\Omega)}=(H(0;g,A,V)u,u)_{\L^2(\Omega)}=
$$
$$
=\int_{\Omega}\<g(x)(-i\nabla_x-A(x))u(x),(-i\nabla_x-A(x))u(x)\>\,dx+\int_{\Omega}V(x)|u(x)|^2\,dx
$$
$$
=\int_{\Omega_\Psi}\<\omega^2(y)(-i\nabla_y-A_*(y))\psi_*(y)v(y),(-i\nabla_y-A_*(y))\psi_*(y)v(y)\>\,dy+
\int_{\Omega_\Psi}V_*(y)\psi_*(y)^2|v(y)|^2\,dy=
$$
$$
=\int_{\Omega_*}\<\omega^2(y)(-i\nabla_y-A_*(y))\psi_*(y)v(y),(-i\nabla_y-A_*(y))\psi_*(y)v(y)\>\,dy+
\int_{\Omega_*}V_*(y)\psi_*(y)^2|v(y)|^2\,dy=
$$
$$
=(H(0;\omega^2\one,A_*,V_*)\psi_*v,\psi_*v)_{\L^2(\Omega_*)}.\,\,\qedhere
$$
\end{proof}
\begin{theorem}
Suppose that $k\in \R^2$. Under the assumptions of Theorem $\ref{main}$, the operator $H(k;g,A,V)$ is unitarily equivalent to the operator
\beq
\label{unitariness}
H\l((T_*^{-1}) ^{t}k,\omega^2\psi_*^{2}\one,A_*,\psi_*^{2}V_*+\psi_*^{2}\omega\Delta\omega-\psi_*\omega\Delta(\psi_*\omega)\r)
\eeq
acting in $\L^2(\Omega_*)$, where $\Omega_*\subset \R^2$ is an elementary cell of $\Gamma_*$, and $T_*, \omega,\psi_*, A_*, V_*$ are defined above.
\end{theorem}
\begin{proof}
We will perform the required unitary transformation in several steps.
First, let us note that $H(k;g,A,V)=H(0;g,A-k,V)$. Consider the unitary transformation $u(x)=e^{i\alpha(x)}v(x)$, where $\alpha\in \c^1_{\per}(\Omega)$. Under this transformation, the operator $H(k;g,A,V)$ becomes
$H(0;g,A-k-\nabla\alpha,V)$. Take $\alpha(x)=k(T_*^{-1}\Psi(x)-x)$. This function is $\Gamma$-periodic, and
$$
(\nabla\alpha)(x)=\Psi'(x)^t (T_*^{-1})^t k-k.
$$
Hence, the operator $H(k;g,A,V)$ is unitarily equivalent to $H(0,g,A-\Psi'(x)^t(T_*^{-1})^t k,V)$, which, by Lemma \ref{changevariables},
is equivalent to
$$
\psi_*H(0,\omega^2\one,A_*-(T_*^{-1})^{t}k,V_*)\psi_*=\psi_*H((T_*^{-1})^{t}k,\omega^2\one, A_*,V_*)\psi_*.
$$
Applying \eqref{521} and then \eqref{522}, we ultimately obtain
$$
\psi_*H((T_*^{-1})^{t}k,\omega^2\one, A_*,V_*)\psi_*=\omega\psi_*H((T_*^{-1})^{t}k,\one,A_*,\omega^{-2}V_*+\omega^{-1}\Delta\omega)\omega\psi_*=
$$
$$
=H((T_*^{-1})^{t}k,\omega^2\psi_*^2\one,A_*,\psi_*^2 V_*+\psi_*^2\omega\Delta\omega-\psi_*\omega\Delta(\psi_*\omega)).\,\qedhere
$$
\end{proof}
This completes the proof of Theorem \ref{main}, because its statement has already been established for the operators \eqref{unitariness}, and the operator families $H(k;g,A,V)$ and \eqref{unitariness} have the same band functions up to a linear transformation of $k$.
\section{An example of degenerate band edge in the discrete case}
Consider the discrete Schr\"odinger operator $H = \D + V$ in $l^2(\Z^2)$, where
$$
(\D u)_n = \frac12 \left(u_{n+e_1} + u_{n-e_1} + u_{n+e_2} + u_{n-e_2}\right),
\quad n \in \Z^2,
$$
is the discrete Laplace operator, and $V$ is the operator of multiplication by the potential given by
$$
(Vu)_n = \begin{cases}
v_0 u_n, \quad \text{if}\ (n_1+n_2)\ \text{is even}, \\
v_1 u_n, \quad \text{if}\ (n_1+n_2)\ \text{is odd}; \end{cases}
$$
the real numbers $v_0$ and $v_1$ are fixed. In other words, the lattice is formed by two different types of atoms placed in a chessboard order, and $V$ is periodic with respect to the lattice spanned by $\{2e_1,e_1+e_2\}$.
The corresponding Floquet-Bloch transform
$$
F : l^2 (\Z^2) \to \L^2 (\tilde \O \times \{0;1\})
$$
is given by
$$
(F u) (k; m) = \frac1{\pi\sqrt2} \sum_{n_1+n_2 \equiv m
    (\operatorname{mod} 2)} e^{-ikn} u_n .
$$
Here
$k \in \tilde \O = \{k \in \R^2 : |k_1+k_2| < \pi\}$,
$m=0$ or $m=1$; the operator $F$ is unitary.
It is easy to see that
$$
F H F^* = \int^\oplus_{\tilde \O} H(k)\, dk,
$$
where $H(k)$ is a self-adjoint operator in $\C^2$,
$$
H(k) =
\left( \begin{array}{cc}
v_0 & \cos k_1 + \cos k_2 \\
\cos k_1 + \cos k_2 & v_1 \end{array} \right) .
$$
Eigenvalues of this matrix are
$$
\la_\pm (k) = \frac{v_0+v_1}2 \pm
\sqrt{\left(\frac{v_0-v_1}2\right)^2 + (\cos k_1 + \cos k_2)^2} \ ,
$$
from which it follows that
$$
\min \la_- = \frac{v_0+v_1}2 - \sqrt{\left(\frac{v_0-v_1}2\right)^2 + 4},
\quad \max \la_- = \min (v_0, v_1),
$$
$$
\min \la_+ = \max (v_0, v_1), \quad
\max \la_+ = \frac{v_0+v_1}2 + \sqrt{\left(\frac{v_0-v_1}2\right)^2 + 4} .
$$
So, the spectrum of the operator $H$ consists of two bands
separated by a gap, whenever $v_0 \neq v_1$.
\begin{figure}[H]
    \includegraphics[scale=0.4]{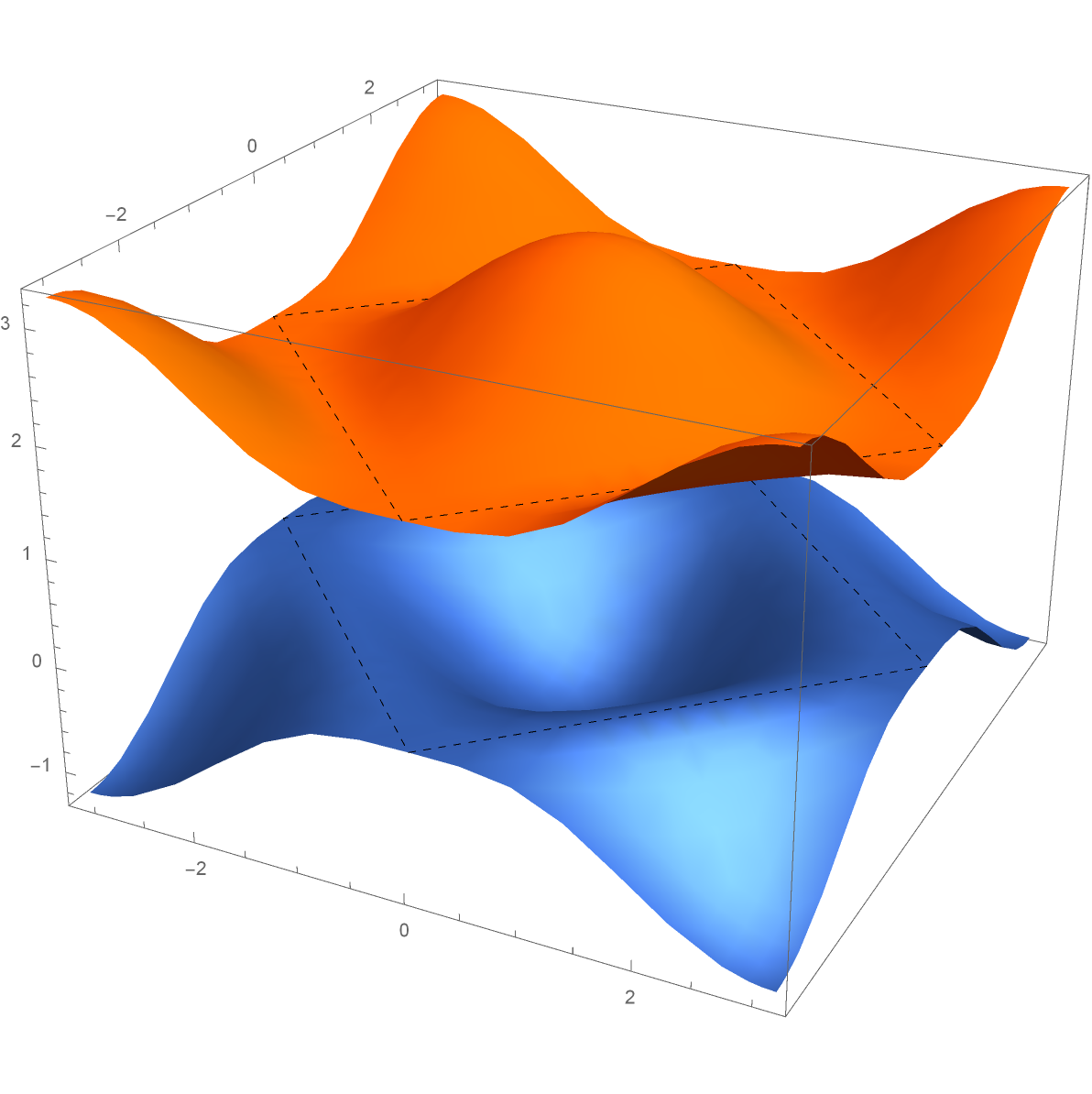}
    \caption{The band functions $\lambda_+(k)$ and $\lambda_-(k)$.}
    \label{fig2}
\end{figure}

The edges of this gap ($v_0$ and $v_1$ respectively)
are attained on the set
$$
\left\{ k \in \R^2 : \cos k_1 + \cos k_2 = 0 \right\}
= \left\{ k \in \R^2 : k_1 \pm k_2 = (2p+1) \pi \right\}_{p\in \Z} ,
$$
which is a countable union of straight lines. Figure \ref{fig2} shows the graphs of $\lambda_{\pm}(\cdot)$ for $v_0=0$, $v_1=2$, with the dashed lines indicating the level sets at the edges of the gap $[0,2]$.
\begin{remark}
	This example seems to be one of the simplest possible 2D diatomic tight binding models. We believe that it should be known to the experts in solid state physics. We could not, however, find it in the literature, which is the reason why we discuss it in detail.
\end{remark}

\end{document}